\documentclass[10pt,twocolumn,twoside]{ieeeconf}
\usepackage{graphicx}
\usepackage{url}
\usepackage{xcolor}
\usepackage{amsmath}
\usepackage{mathtools}
\usepackage{tikz}
\usepackage{verbatim}
\usepackage{subcaption}
\usepackage{pgfplots}
\pgfplotsset{compat=1.10}
\usepackage{graphicx}
\usepackage{epstopdf}
\usepackage{amssymb}
\usepackage{relsize}
\usepackage[textwidth=3.8em,textsize=scriptsize,disable]{todonotes}
\usepackage{algorithm}
\definecolor{Gray}{gray}{0.90}
\usepackage{colortbl}

\usepackage{algorithm}
\usepackage{algorithmicx}
\usepackage{algcompatible}
\usepackage{algpseudocode}
\usepackage{color}
\usepackage{blkarray}

\newtheorem{theorem}{\bf{Theorem}}[section]

\newtheorem{lem}[theorem]{\bf{Lemma}}
\newtheorem{prop}[theorem]{\bf{Proposition}}

\newtheorem{fact}[theorem]{Fact}

\newenvironment{definition}[1][Definition]{\begin{trivlist}
\item[\hskip \labelsep {\bfseries #1}]}{\end{trivlist}}
\def\qed{\hfill\rule[-1pt]{5pt}{5pt}\par\medskip}
\usepackage{amssymb}
\usepackage{url}
\newcommand{\calG}[0]{\mathcal{G}}
\newcommand{\calV}[0]{\mathcal{V}}
\newcommand{\calE}[0]{\mathcal{E}}

\newcommand{\calA}[0]{\mathcal{A}}
\newcommand{\calB}[0]{\mathcal{B}}

\newcommand{\calS}[0]{\mathcal{S}}
\newcommand{\calN}[0]{\mathcal{N}}

\newcommand{\calH}[0]{\mathcal{H}}
\newcommand{\calD}[0]{\mathcal{D}}

\newcommand{\calL}[0]{\mathcal{L}}
\definecolor{amber}{rgb}{1.0, 0.75, 0.0}
\begin{document}

\title{Improving Network Robustness through Edge Augmentation While Preserving Strong Structural Controllability}
\author{Waseem~Abbas, Mudassir~Shabbir, Hassan~Jaleel, and Xenofon~Koutsoukos
\thanks{This work was supported in part by the National Institute of Standards and Technology under Grant 70NANB18H198.}
\thanks{W.~Abbas and X.~Koutsoukos are with the Electrical Engineering and Computer Science Department at the Vanderbilt University, Nashville, TN, USA (Emails: waseem.abbas@vanderbilt.edu, xenofon.koutsoukos@vanderbilt.edu).}
\thanks{M.~Shabbir is with the Computer Science Department at the Information Technology University, Lahore, Pakistan (Email: mudassir@rutgers.edu).}
\thanks{H.~Jaleel is with the Electrical Engineering Department at Lahore University of Management Sciences, Lahore, Pakistan (Email: hassan.jaleel@lums.edu.pk).}
}

\maketitle
\begin{abstract}
In this paper, we consider a network of agents with Laplacian dynamics, and study the problem of improving network robustness by adding a maximum number of edges within the network while preserving a lower bound on its strong structural controllability (SSC) at the same time. Edge augmentation increases network's robustness to noise and structural changes, however, it could also deteriorate network controllability. Thus, by exploiting relationship between network controllability and distances between nodes in graphs, we formulate an edge augmentation problem with a constraint to preserve distances between certain node pairs, which in turn guarantees that a lower bound on SSC is maintained even after adding edges. In this direction, first we choose a node pair and maximally add edges while maintaining the distance between selected nodes. We show that an optimal solution belongs to a certain class of graphs called clique chains. Then, we present an algorithm to add edges while preserving distances between a certain collection of nodes. Further, we present a randomized algorithm that guarantees a desired approximation ratio with high probability to solve the edge augmentation problem.  Finally, we evaluate our results on various networks.
\end{abstract}

\section{Introduction}
Network controllability has been an active research topic in the broad domain of systems and control as well as in network science in recent years \cite{pasqualetti2014controllability}. The main goal is to understand how can we manipulate a network of dynamical agents, often represented by a directed or an undirected graph, by controlling only a small subset of agents, referred to as leaders. In a networked dynamical system, the underlying network topology significantly influences its controllability. Therefore, it is crucial to develop a topological view-point of network controllability \cite{chapman2013strong,monshizadeh2014zero,aguilar2015graph,yaziciouglu2016graph,mousavi2016controllability,van2017distance}. Another important aspect here is to consider the effect of weights given by nodes to each other's information, which is typically modeled by assigning edge weights in the graph. In fact, we say that a network is strong structurally controllable if it is possible to control the entire network with an arbitrary choice of non-zero edge weights. This controllability notion, which is independent of edge weights is particularly useful when exact coupling strengths between nodes are unknown.

Along with the network controllability, we also desire to improve other attributes of a networked dynamical system, in particular, its robustness to changes in the underlying network topology as well as to the noisy information. A network can be made robust by adding more links (edges) between nodes. For instance, a widely used measure of network robustness is Kirchhoff index $K_f$, which is simply the sum of the reciprocal of the (non-zero) eigenvalues of the graph Laplacian. It measures the effect of structural changes in the network topology as well as the effect of noise on the overall dynamics \cite{young2010robustness,ellens2011effective,abbas2012robust,zelazo2015robustness,YasinCDC19}. It is well known that adding edges always improves network's robustness as measured by $K_f$. In fact, network robustness increases monotonically with edge additions.

Adding edges and densifying a graph is effective on the one hand as it improves network's robustness, but on the other hand it can also deteriorate network controllability \cite{Abbas2019,pasqualetti2018fragility}. For instance, we have shown in \cite{Abbas2019} that for given network parameters, such as the number of nodes and the diameter, networks with maximal robustness require a large number of input nodes (leaders) to become controllable. Similarly networks that are controllable with minimum leaders are least robust to noise and to structural changes. Thus, it is important to identify edges whose addition to the network minimally reduces its controllability.

In this paper, we study the problem of maximally adding edges in a network while preserving its strong structural controllability (SSC). It is easy to verify, in fact in a linear time, whether the entire network is strong structurally controllable \cite{weber2014linear}.  If it is not, computing exactly how much of the network is indeed strong structurally controllable is an extremely challenging problem. Thus, lower bounds on the SSC of networks have been studied in the literature \cite{chapman2013strong,monshizadeh2015strong,yaziciouglu2016graph,zhang2014upper,monshizadeh2014zero,mousavi2018structural,van2017distance}. In this work, we utilize a tight lower bound on SSC, which is based on the topological distances between nodes, and propose algorithms to densify the original network while preserving this lower bound. Our main contributions are:

\begin{itemize}
\item We formulate the problem of adding maximal edges in a network while preserving a lower bound on SSC as an edge augmentation problem in which distances between certain node pairs in a graph need to be maintained. 

\item We show that for a fixed node pair, optimal solution to the distance preserving edge augmentation problem belongs to a class of graphs known as clique chains.

\item We provide two algorithms, including a randomized algorithm, to add edges in a graph while preserving distances between a collection of node pairs, thus solving the problem of adding edges while ensuring that a lower bound on SSC is maintained. We also show that the proposed randomized algorithm is an approximation algorithm with high probability.


\item Finally, we numerically evaluate our results on various networks.  
\end{itemize}

Numerous results are available in the context of graph sparsification, where the objective is to remove edges while preserving distances between nodes exactly or approximately (e.g., see \cite{bernstein2019distance,baswana2007simple,bollobas2005sparse,spielman2011graph}). However, to the best of our knowledge, this paper is novel in considering the opposite problem, that is the densification of graphs while preserving distances between nodes, and then applying it towards preserving SSC in networks. We mention that a recent paper \cite{mousavi2017robust} studies the problem of characterizing edges, which if added to a directed network preserve its SSC. However, results hold if the entire network is strong structurally controllable. In our case, we solve the maximal edge addition problem in networks even if they are only partially strong structurally controllable. Further studies in the domain of network controllability include selecting inputs and leader nodes to structurally control a network, for instance see \cite{pequito2015complexity,tzoumas2015minimal,pequito2015framework,olshevsky2014minimal,yaziciouglu2013leader,clark2017submodularity}.

The rest of the paper is organized as below: Section \ref{sec:prelim} introduces preliminaries and formulates the problem. Section \ref{sec:DEPA_node_pair} presents an overview of our approach along with the distance preserving edge augmentation problem. Section \ref{sec:DEPA_SSC} provides algorithms to add edges in a network while preserving SSC. Section \ref{sec:evaluation} presents an evaluation of our results, and Section \ref{sec:Con} concludes the paper.

\section{Preliminaries and Problem Description}
\label{sec:prelim}
\subsection{Notations}
We consider a network of dynamical agents modeled by a simple undirected graph $\calG=(\calV,\calE)$, in which the node set represents agents and the edge set represents interconnections between agents. Nodes\footnote{We use terms node and vertex interchangeably throughout the paper.} $a$ and $b$ are \emph{adjacent} if there is an edge between them. The \emph{neighborhood} of node $a$, denoted by $\calN_a$ is the set of nodes adjacent to $a$, and the number of nodes in $\calN_a$ is the \emph{degree} of $a$. The distance between nodes $a$ and $b$ in $\calG$, denoted by $d_\calG(a,b)$, is the number of edges in the shortest path between $a$ and $b$. The maximum distance between any two nodes in a graph is called its \emph{diameter}. Moreover, edges can be weighted by some weight function, $w:\calE\mapsto \mathbb{R}_+$. The edge weight represents the coupling strength between nodes.  

Agent $a$ has a state $x_a\in\mathbb{R}$. Without loss of generality, the overall state of the network is $x\in\mathbb{R}^n$ where $n = |\calV|$. Agents follow the Laplacian dynamics given by

\begin{equation}
\label{eq:dynamics}
\dot{x} = -\calL_w\ x + \calB u.
\end{equation}

Here, $\calL_w\in\mathbb{R}^{n\times n}$ is a weighted Laplacian matrix of the graph $\calG$, and is defined as $\calL_w = -\calA + \Delta$, where $\calA$ is the weighted adjacency matrix in which the $ij^{th}$ entry, that is $\calA_{ij}$ is simply the weight on the edge between nodes $i$ and $j$ if such an edge exists, and is zero otherwise. $\Delta$ is the diagonal matrix whose $i^{th}$ diagonal entry is simply $\sum_{j=1}^{n} \calA_{ij}$. At the same time $\calB\in\mathbb{R}^{n\times m}$ is an input matrix, where $m$ is the number of inputs, or simply the number of leaders---nodes with an external control signal. Let $\calV_\ell =~ \{\ell_1,\ell_2,\cdots,\ell_m\} \subset ~\calV$ be the set of leaders, then $\calB_{i,j} = 1$ if node $i\in\calV$ is also a leader $\ell_j$, otherwise $\calB_{i,j} = 0$. 

\subsection{Strong Structural Controllability (SSC)} 
A state $x_f\in\mathbb{R}^n$ is reachable if an input exists that can drive the network in \eqref{eq:dynamics} from an initial state $x_0~=~\left[0 \; 0 \; \cdots 0\right]^T\in\mathbb{R}^n$ to $x_f$ in a finite amount of time. A network is \emph{completely controllable} if every point in $\mathbb{R}^n$ is reachable. Complete controllability of a network $\calG(\calV,\calE)$ with given edge weights $w$ and leaders $\calV_\ell$ can be checked by computing the rank of the following controllability matrix.

\begin{equation*}
\small
\Gamma(\calL_w,\calB) = 
\left[
\begin{array}{ccccc}
\calB & (-\calL_w)\calB & (-\calL_w)^2\calB & \cdots & (-\calL_w)^{n-1}\calB
\end{array}
\right].
\end{equation*}

The network is completely controllable if and only if $\text{rank}(\Gamma(\calL_w,\calB)) = n$, and the pair $(\calL_w,\calB)$ is called the controllable pair. In fact, the rank of $\Gamma$ defines the dimension of controllable subspace consisting of all the reachable states.

In a network, if we fix the leaders ($\calB$) and the edge set ($\calE$), the rank of controllability matrix can change with a different choice of edge weights $w'$, that is the $\text{rank}(\Gamma(\calL_w,\calB))$ might be different from the $\text{rank}(\Gamma(\calL_{w'},\calB))$. A network $\calG=~(\cal(V,\calE)$ is called \emph{strong structurally controllable} with a given leader set $\calV_\ell$ if it is completely controllable for any choice of edge weights, or in other words $(\calL_w,\calB)$ is a controllable pair for any choice of $w$. At the same time, the \emph{dimension of strong structurally controllable subspace}, or simply the \emph{dimension of SSC} is the minimum rank of the controllablility matrix $\Gamma(\calL_w,\calB)$ over all possible edge weights $w$. Roughly, the dimension of SSC quantifies how much of the network can be controlled by a given leader set with arbitrary edge weights. 

\subsection{Problem Description}
Our goal is to add a maximum number of edges within the network while preserving the dimension of its strong structurally controllable subspace at the same time. However, computing the exact dimension of SSC with a given set of leaders is a computationally hard problem. As a result, finding good lower bounds on the dimension of such a subspace has been an active research topic. Our approach is to select a tight lower bound, and then add maximal edges within the network while preserving the lower bound on the dimension of SSC. We discuss the lower bound used in this work in the next subsection.
\begin{definition}[Problem]
\label{def:Problem}
Consider a network of agents $\calG=(\calV,\calE)$ with the dynamics in \eqref{eq:dynamics}. Let $\calV_\ell\subset\calV$ be the leaders and the dimension of SSC of $\calG$ with $\calV_\ell$ is at least $\delta$. Then, our task is to find a maximum size edge set $\calE'$ such that $\calE\subseteq\calE'$ and the dimension of SSC of the network induced by $\calV$ and $\calE'$, say $\calH=(\calV,\calE')$, is also at least $\delta$ with the same leaders $\calV_\ell$. 
\end{definition}
\subsection{A Tight Lower Bound on SSC Based on Graph Distances}
We utilize a lower bound proposed in \cite{yaziciouglu2016graph} that is based on the distances between nodes in a graph. Assuming $m$ leaders $\calV_\ell = \{\ell_1,\cdots,\ell_m\}$, we define a \emph{distance-to-leader} vector for each node $a\in\calV$ in $\calG$ as below,

\begin{equation*}
\label{eq:DLvector}
D_a = \left[
\begin{array}{ccccc}
d_\calG(\ell_1,a) & d_\calG(\ell_2,a) & \cdots & d_\calG(\ell_m,a)
\end{array}
\right]^T \in \mathbb{Z}_+^m.
\end{equation*}

The $j^{th}$ component of $D_a$, denoted by $D_{a,j}$, is the distance between node $a$ and leader $j$, that is $D_{a,j} = d_\calG(\ell_j,a)$. Next, we define a sequence of such vectors, called as \emph{pseudo-monotonically increasing} sequence as below.

\begin{definition}{(\emph{Pseudo-monotonically Increasing Sequence (PMI)}}
\label{def:PMI}
A sequence of distance-to-leader vectors $\calD$ is PMI if for every $i^{th}$ vector in the sequence, denoted by $\calD_i$, there exists some $\alpha(i)\in\{1,2,\cdots,m\}$ such that 
\begin{equation}
\label{eq:PMIcondition}
\calD_{i,\alpha(i)} < \calD_{j,\alpha(i)}, \;\;\forall j > i.
\end{equation}
\end{definition}

An example of distance-to-leader vectors is illustrated in Figure \ref{fig:PMI}(a). A PMI sequence of length five is 

\begin{equation*}
\label{eq:PMIexample}
\footnotesize
\calD = 
\left[
\left[
\begin{array}{c}
\textcircled{0} \\ 2
\end{array}
\right]
,
\left[
\begin{array}{c}
2 \\ \textcircled{0}
\end{array}
\right]
,
\left[
\begin{array}{c}
\textcircled{1} \\ 2
\end{array}
\right]
,
\left[
\begin{array}{c}
\textcircled{2} \\ 1
\end{array}
\right]
,
\left[
\begin{array}{c}
\textcircled{3} \\ 1
\end{array}
\right]
\right].
\end{equation*}

Note that for each vector in the sequence, there is an index---of the circled value---such that the values of all the subsequent vectors at the corresponding index are strictly greater than the circled value, thus satisfying the condition in \eqref{eq:PMIcondition}. 

\begin{figure}[htb]
\centering
\includegraphics[scale=0.65]{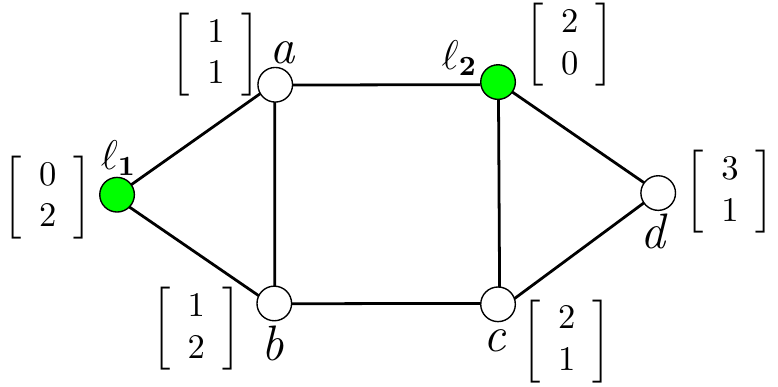}
\caption{A network with two leaders $\calV_\ell = \{\ell_1,\ell_2\}$. The distance-to-leader vectors of all nodes are also shown. A PMI sequence of length five is $\calD = [\calD_1 \; \calD_2 \; \cdots \; \calD_5] = [D_{\ell_1} \; D_{\ell_2} \; D_b \; D_c \; D_d]$.}
\label{fig:PMI}
\end{figure}

The length of PMI sequence of distance-to-leader vectors is related to strong structural controllability as below.

\begin{theorem} \cite{yaziciouglu2016graph}
\label{thm:PMI}
If $\delta$ is the length of a longest PMI sequence of distance-to-leader vectors in a network, then the dimension of SSC of the network is at least $\delta$.
\end{theorem}

 In \cite{MudassirCDC19}, we presented a dynamic programming based exact algorithm and an approximate greedy algorithm that returns near optimal PMI sequence of distance-to-leader vectors in $O(mn\log n)$ time where $m$ is the number of leaders and $n$ is the total number of nodes in a graph.
\section{Edge Augmentation While Preserving Distances in Graphs}
\label{sec:DEPA_node_pair}

Our approach is to add a maximal edge set to the original graph while ensuring that the maximum lengths of PMI sequences\footnote{For brevity, we use the term `PMI sequence' instead of `PMI sequence of distance-to-leader vectors'.} of the resulting and the original graphs remain the same. As a result, we preserve a lower bound on the dimension of SSC in the original graph even after adding edges to it. We note that the bound based on the maximum length of PMI sequence (Theorem \ref{thm:PMI}) is tight and numerical evaluation in \cite{MudassirCDC19} shows that it generally performs better than the other known bounds, such as the one based on zero forcing sets \cite{monshizadeh2015strong}. Moreover, we provide algorithms to efficiently compute maximum length PMI sequence in \cite{MudassirCDC19}. 

We proceed by letting $\calG=(\calV,\calE)$ and $\calH=(\calV,\calE')$ respectively denote the original graph and the graph obtained after adding edges, where $\calE\subseteq\calE'$. 
At the same time, consider $\calD$ to be a PMI sequence of maximum length in $\calG$ with $\calV_\ell$ leaders. If $\calD$ is of length $\delta$, then by Theorem \ref{thm:PMI}, the dimension of SSC of $\calG$ is at least $\delta$. At the same time, if $\calH$ is such that $d_\calG(\ell,a) = d_\calH(\ell,a)$, where $a$ is any such node whose distance-to-leader vector is included in $\calD$ and $\ell$ is an arbitrary leader in $\calV_\ell$, then $\calD$ is also a PMI sequence in $\calH$. Consequently, the dimension of SSC of $\calH$ is also lower bounded by $\delta$. 
Thus, to add edges while preserving a lower bound on the dimension of SSC, our approach is to add edges while ensuring that the distance between leaders and a certain subset of nodes (whose distance-to-leader vectors are included in a maximum length PMI sequence of $\calG$) are preserved. In fact, if there are $m$ leaders, and the maximum length PMI sequence is $|\calD| = \delta$, then we need to preserve distances between $\frac{m(m-1)}{2} + m(\delta - m)$ node pairs. Thus, the problem of edge addition while preserving strong structural controllability becomes the edge augmentation problem in networks while preserving distances between nodes.

\subsection{Adding Edges While Preserving Distance Between Two Nodes}
\label{sec:DPEA}
We proceed by fixing a node pair $a,b\in\calV$ and finding a maximal edge set, which if added to the original graph, preserves the distance between $a$ and $b$. We call this as the \emph{Distance Preserving Edge Augmentation (DPEA) problem}, formally stated below. We solve the DPEA problem for all the node pairs $v,\ell$ where $\ell$ is a leader and $v$ is a node whose distance-to-leader vector is included in a maximum length PMI sequence $\calD$. Taking an intersection of solutions to all the instances of DPEA problem then gives a maximal edge set that can be added to the original graph $\calG$ to obtain a new graph $\calH$ in which distances between leaders and nodes in the PMI sequence $\calD$ are preserved. As a result, $\calD$ is also a PMI sequence in $\calH$ and a lower bound on the dimension of SSC of $\calG$ also holds for $\calH$. Next, we formulate and solve the DPEA problem.

\begin{definition}[Distance Preserving Edge Augmentation (DPEA)]
Given an undirected graph $\mathcal{G} = (\mathcal{V},\mathcal{E})$ and $a,b\in\mathcal{V}$ such that $d_\calG(a,b)=k$, we are interested in $\mathcal{H}=(\mathcal{V},\mathcal{E'})$ where $\mathcal{E}\subseteq \mathcal{E}'$, such that $d_\calG(a,b)=d_\calH(a,b)=k$. Our goal is to find $\mathcal{H}=(\mathcal{V},\mathcal{E'})$ with the maximum $|\mathcal{E}'|$.
\end{definition}

We show that an optimal solution to the DPEA problem for a given pair of nodes belongs to a class of graphs known as \emph{clique chains} \cite{ellens2011effective}, which we define below.

\begin{definition}{(\emph{Clique chain)}}
\label{def:CliqueChains}
Let $n_0,n_1,\cdots,n_{k} \in \mathbb{Z}_+$ and $n=\sum_{i=0}^{k}n_i$, then a clique chain of $n$ nodes is a graph obtained from a path graph of length $k$ by replacing each node with a clique\footnote{All vertices in a clique are pair-wise adjacent.} of size $n_i$ such that the vertices in distinct cliques are adjacent if and only if the corresponding vertices in the path graph are adjacent. We denote such a clique chain by $\mathcal{G}_k(n_0,\cdots,n_{k})$.
\end{definition}

An example of a clique chain is illustrated in Figure \ref{fig:CC}. Note that the diameter of $\mathcal{G}_k(n_0,\cdots,n_{k})$ is $k$.

\begin{figure}[h]
\centering
\includegraphics[scale=0.7]{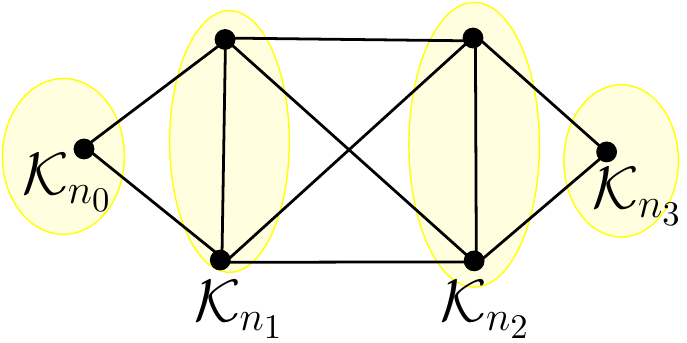}
\caption{A clique chain $\calG_3(1,2,2,1)$ with $n_0 = 1$, $n_1 = 2$, $n_2 = 2$ and $n_3 = 1$.}
\label{fig:CC}
\end{figure}

\begin{theorem}
\label{thm:clique_chain_optimal}
For a given $\cal{G}=(\calV,\calE)$, and $a,b\in\calV$ where $d_{\mathcal{G}}(a,b) > 1$, optimal solution to the DPEA problem is a clique chain of the form $\calG_k(n_0 = 1, n_1,\cdots,n_{k-1},n_{k} = 1)$, where $\sum_{i=0}^{k}n_i =~|\calV|$.
\end{theorem}

\emph{Proof:} Let $\cal{G}^*=(\calV,\calE^*)$ be an optimal solution to the DPEA problem for a given $\cal{G}=(\calV,\calE)$, and $a,b\in\calV$. Also let $d_{\cal{G}}(a,b) = k > 1$, and note that for every $v\in \calV$ the following holds

\begin{equation}
d_{\cal{G}^*}(a,v) + d_{\cal{G}^*}(v,b)  = k.
\label{eq:distance_sum}
\end{equation}

Clearly the sum $d_{\cal{G}^*}(a,v) + d_{\cal{G}^*}(v,b)$ can not be less than $k$ since $\cal{G}^*$ preserves the distance between $a,b$. If it is more than $k$ then at least one edge can be added between $v$ and a node that is at a distance two from $v$, which contradicts the optimality of $\cal{G}^*$. For $0\le i\le k$ let's define levels as
$$ 
S_i := \{ v\in\calV: d_{\cal{G}^*}(a,v) = i  \},
$$
that is, nodes in $S_i$ are at distance $i$ from $a$; also they are at distance $k-i$ from $b$. This defines a fixed partition of all nodes. Observe that for a pair of nodes $x\in S_i$ and $y\in S_j$ where $j-i>1$, we can not have an edge between $x$ and $y$ in $\cal{G}^*$ as otherwise it will violate \eqref{eq:distance_sum}. A clique chain defined on levels $S_0,S_1\ldots,S_k$, that is $\calG_k(|S_0|,|S_1|,\cdots,|S_{k-1}|,|S_k|)$ contains all remaining edges, where $|S_0|=|S_k|=1$, which concludes our proof. \qed

Note that when $d_{\mathcal{G}}(a,b)=1$, optimal solution to the DPEA problem is trivially a complete graph. Next, we construct a clique chain for a given pair of nodes $a,b$. Note that the original graph $\calG$ must be a subgraph of this clique chain.
\subsection{Clique Chain Construction}
To construct a clique chain, we define $S_i^a$ to be the set of nodes that are at a distance $i$ from $a$ and similarly define $S_i^b$. If $d_\calG(a,b) = k$ is odd, then the sets $S_i^v$ where $i\in\lbrace{0,1,\cdots,{\lfloor{k/2}\rfloor}\rbrace}$ and $v\in\{a,b\}$ are all non-empty and are pairwise disjoint. Similarly, for even $k$, the sets $S_i^a$ where  $i\in\lbrace{0,1, \cdots, k/2\rbrace}$ and $S_i^b$ where $i\in\lbrace{0,1,\cdots, \frac{k}{2}- 1\rbrace}$ are all pairwise disjoint and non-empty. Next, we define free and fixed nodes as below:

\begin{definition}{(\emph{Fixed and free nodes)}}
\label{def:fixed}
For a given pair of nodes $a,b\in\calV$, all such nodes that are included in some \emph{shortest} path between $a$ and $b$ are referred to as \emph{fixed} nodes, while the remaining nodes are called \emph{free} nodes.
\end{definition}

We note that every fixed node must lie in $S_i^v$, where $v\in~\{a,b\}$. 
However, there could be free nodes that might not be included in any $S_i^a$ or $S_i^b$. For instance, if $k$ is even, consider a node $x$ with $d_\calG(x,a)> k/2$ and $d_\calG(x,b) > \frac{k}{2}-1$; and if $k$ is odd, consider $x$ such that $d_\calG(x,a)> \lfloor k/2\rfloor$ and $d_\calG(x,b) > \lfloor k/2\rfloor$. We can always include such a free node $x$ into the set $S_{\lfloor k/2 \rfloor}^a$ by adding an edge between $x$ and $y$ for some $y \in S_{\lfloor k/2 \rfloor - 1}^a$ while preserving the distance $k$ between nodes $a$ and $b$. Furthermore, if a free node is already included in some $S_i^a$ (or $S_i^b$), it might be possible to place it in some $S_j^a$ (or $S_j^b$) for some $j\ne i$ by creating an edge between $x$ and some $y\in S_{j-1}^a$ (or $y\in S_{j-1}^b$) as long as the distance between $a$ and $b$ is maintained. However, if $x$ is a fixed node in some $S_i^a$ (or $S_i^b$), then it can never be placed in $S_j^a$ (or $S_j^b$) for any $j\ne i$ without changing the distance between $a$ and $b$. Moreover, each of the $S_i^a$ and $S_i^b$ contains at least one fixed node. As a result, for given $a$ and $b$ in $\calG=(\calV,\calE)$, we always have a partition of $\calV$ into $k+1$ subsets given by,

\begin{equation}
\label{eq:odd}
\calS = \lbrace{S_0^a,S_1^a,\cdots,S_{\lfloor \frac{k}{2} \rfloor}^a,S_{\lfloor \frac{k}{2}\rfloor}^b, \cdots, S_1^{b},S_0^b\rbrace},\;\;\;\text{(odd } k);
\end{equation}
and
\begin{equation}
\label{eq:even}
\calS = \lbrace{S_0^a,S_1^a,\cdots,S_{k/2}^a,S_{\frac{k}{2} - 1}^b, \cdots, S_1^b,S_0^b\rbrace}\;\;\;\text{(even } k), 
\end{equation}
where $S_0^a = \{a\}$ and $S_0^b = \{b\}$.

Figures \ref{fig:ideas}(a)--(c) illustrate above notions. Next, we induce a clique chain over subsets in $\calS$, that is $\calG_k(1,|S^a_1|,|S^a_2|,\cdots,|S^b_2|,|S^b_1|,1)$. The distance between $a$ and $b$ in this clique chain is $k$, and it contains the original graph $\calG$ as a subgraph. We illustrate this in Figure \ref{fig:ideas}(d).

\begin{figure*}
\centering
\includegraphics[scale=0.8]{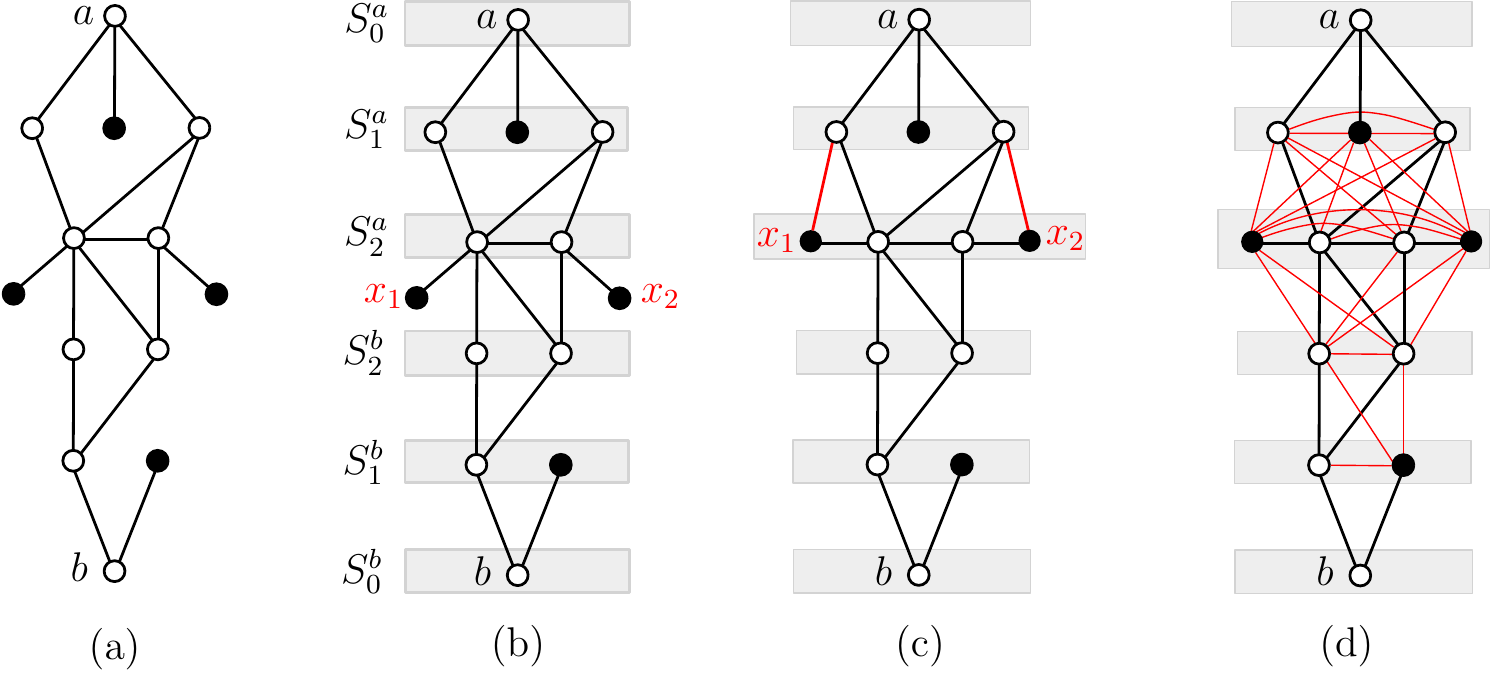}
\caption{All white nodes are \emph{fixed} nodes as each of them lie on some shortest path between $a$ and $b$, whereas black nodes do not lie on any shortest path between $a$ and $b$ and are \emph{free} nodes. (b) Nodes are partitioned into sets $S_i^v$ where $v\in\{a,b\}$ and $i\in\{0,1,2\}$. Note that free nodes $x_1$ and $x_2$ are not included in any $S_i^a$ or $S_i^b$. (c) Both $x_1$ and $x_2$ are included in $S_2^a$ by creating an edge between $x_1$ and some node in $S_1^a$, and similarly by an edge between $x_2$ and a node in $S_1^a$. The distance between $a$ and $b$ does not change by the addition of these edges. (d) Clique chain is induced over all $S_i^a$ and $S_i^b$.}
\label{fig:ideas}
\end{figure*}

As for the time complexity of constructing such a clique chain, a breadth-first search (BFS) tree starting at $a$ (similarly starting at $b$) gives us distances $d_\calG(a,v)$ (similarly $d_\calG(b,v)$) for all vertices $v$. Once these two distances are available for each node $v$, we can determine the subset in $\calS$ to which $v$ belongs to. Further, we can check if the node is free or fixed by verifying the equation:  $d_\calG(a,b)=d_\calG(a,v)+d_\calG(b,v)$. Consequently, for  a given graph $\calG=(\calV,\calE)$ and a pair of nodes $a,b$, we can construct the desired clique chain in $O(|\calV|+|\calE|)$ time.

\section{Edge Augmentation While Preserving a Lower Bound on SSC}
\label{sec:DEPA_SSC}
In this section, we present two algorithms to add maximal edges in the existing graph while preserving a (PMI based) lower bound on the dimension of SSC. For a graph $\mathcal{G}$, our algorithms take PMI sequence of distance-to-leader vectors $\mathcal{D}$ as input and return edges whose addition to $\mathcal{G}$ does not change distance-to-leader vectors of nodes included in $\calD$. As a result, $\calD$ is also a PMI sequence of the graph even after adding extra edges.
\subsection{Intersection Algorithm}
If $\calD$ is a PMI sequence and $\bar\calV\subseteq\cal{V}$ is the set of nodes whose  distance-to-leader vectors are included in $\calD$, then our goal is to add edges that do not change the distance between nodes in the node pair $(v,\ell)$, where $v\in\bar{\calV}$ and $\ell\in\calV_\ell$. One way of achieving this is to solve DPEA problem (as discussed in Section \ref{sec:DPEA}) for each such node pair, and then select edges that are common in solutions of all such DPEA instances. These common edges, if added to the graph, will preserve the distance between any leader node and $v\in\bar{\calV}$. 
We summarize this scheme in Algorithm \ref{algo:DPEA}.

%
%

\begin{algorithm}
\caption{Intersection Algorithm}
\label{algo:DPEA}
\begin{algorithmic}[1]
\State Consider a PMI sequence $\calD$ in $\calG$ with leaders in $\calV_\ell$.
\State Identify nodes whose distance-to-leader vectors are included in $\calD$, and denote them by $\bar{\calV}\subseteq\calV$.
\State For each node pair $(\ell,v)$ where $\ell\in\calV_\ell$ and $v\in\bar{\calV}$, solve the DPEA problem. Let $\calE_{\ell,v}$ be a solution.
\State Compute 
\begin{equation}
\label{eq:final_e}
\calE' = \bigcap\limits_{\ell\in\calV_\ell;\;v\in\bar{\calV}}\calE_{\ell,v}.
\end{equation}
\end{algorithmic}
\end{algorithm}

\begin{prop}
If $\delta$ is a distance-based lower bound on the dimension of SSC of $\calG=(\calV,\calE)$ with leaders $\calV_\ell\subset\calV$, then $\delta$ is also a lower bound on the dimension of SSC of a graph $\calH=(\calV,\calE')$, where $\calE'$ is given in \eqref{eq:final_e}.
\end{prop}

\emph{Proof} Let $\calD$ be a PMI sequence of length $\delta$ in $\calG=(\calV,\calE)$ containing distance-to-leader vectors of nodes in $\bar{\calV}\subseteq\calV$. Using the above scheme, we compute $\calE'$ and obtain a graph $\calH=(\calV,\calE')$. Note that the distances between leaders and nodes in $\bar{\calV}$ is exactly same in $\calG$ and $\calH$. Consequently, $\calD$ is also a PMI sequence in $\calH$, and hence $\delta$ is also a lower bound on the dimension of SSC of $\calH$. \qed



Since there are at most $|\calV_\ell|\times |\calD|$ instances of the DPEA problem, and each such instance takes $O(|\calV|+|\calE|)$ time, Algorithm \ref{algo:DPEA} runs in $O(|\calV_\ell|\times |\calD|\times (|\calV|+|\calE|))$ time.

\subsubsection*{Example}
As an example, consider the network in Figure~\ref{fig:edge_addition}(a) with $|\calV|=15$ nodes and $|\calE|=23$ edges. We consider the cases with one and two leaders as illustrated in Figures \ref{fig:edge_addition}(b) and \ref{fig:edge_addition}(c) respectively. With a single leader ($\calV_\ell = \{v_1\}$), the dimension of SSC is at least 6, as the maximum length of a PMI sequence is 6 as given below. We can add 41 extra edges (shown in red in Figure \ref{fig:edge_addition}(b)), while ensuring that the distance between the leader node $v_1$ and each of the node in $\bar{\calV} = \{v_2,v_3,v_4,v_5,v_6\}$ is preserved in the new graph. 

\begin{equation*}
\calD = 
\label{eq:one_lead}
\begin{blockarray}{ccccccc}
 v_1 & v_5 & v_4 & v_3 & v_2 & v_6 \\
\begin{block}{[ccccccc]}
 0 & 1 & 2 & 3 & 4 & 5 \\
\end{block}
\end{blockarray}
\end{equation*}

Similarly, with two leaders $\calV_\ell = \{ v_1, v_4\}$, the dimension of SSC is at least 10. A PMI sequence of length 10 is given below along with the nodes whose distance-to-leader vectors are included in the sequence. 

\begin{equation*}
\calD = 
\label{eq:one_lead}
\begin{blockarray}{ccccccccccc}
 v_1 & v_4 & v_5 & v_{9} & v_7 & v_{15} & v_{3} & v_{11} & v_{2} & v_6\\
\begin{block}{[ccccccccccc]}
   \textcircled{0}   &   2  &     \textcircled{1}   &    2   &    \textcircled{2}   &   3  &    \textcircled{3}   &   4  &    \textcircled{4}   &    \textcircled{5}  \\
 2  &    \textcircled{0}  &    1   &   \textcircled{1}     & 2   &   \textcircled{2}    &  3  &    \textcircled{3}   &  4     & 5\\
\end{block}
\end{blockarray}
\end{equation*}
We can add 21 extra edges in the graph while ensuring that above sequence is still a PMI sequence in the new graph.
\begin{figure}[htb]
\centering
    \begin{subfigure}[b]{0.16\textwidth}
	\centering
	\includegraphics[scale=0.27]{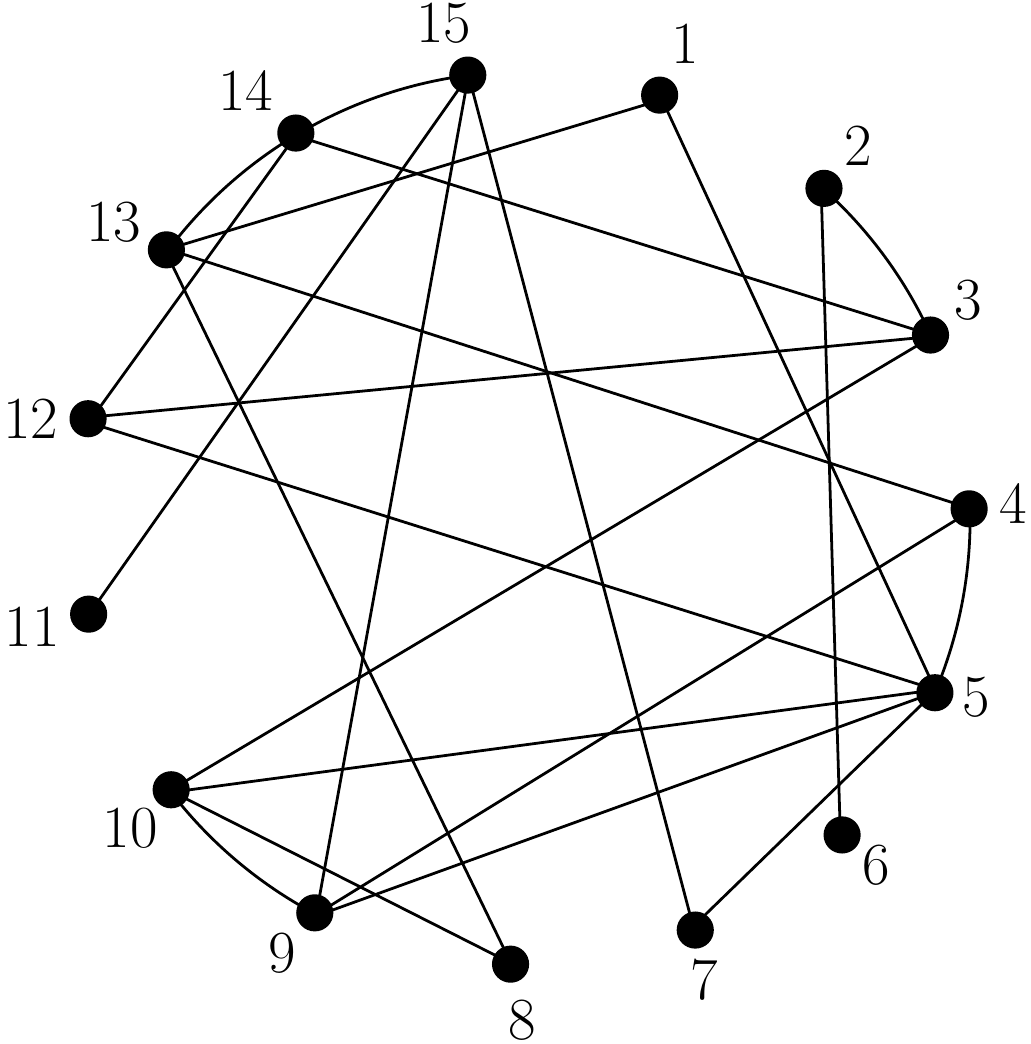}
	\caption{}
	\end{subfigure}
	\begin{subfigure}[b]{0.16\textwidth}
	\centering
	\includegraphics[scale=0.27]{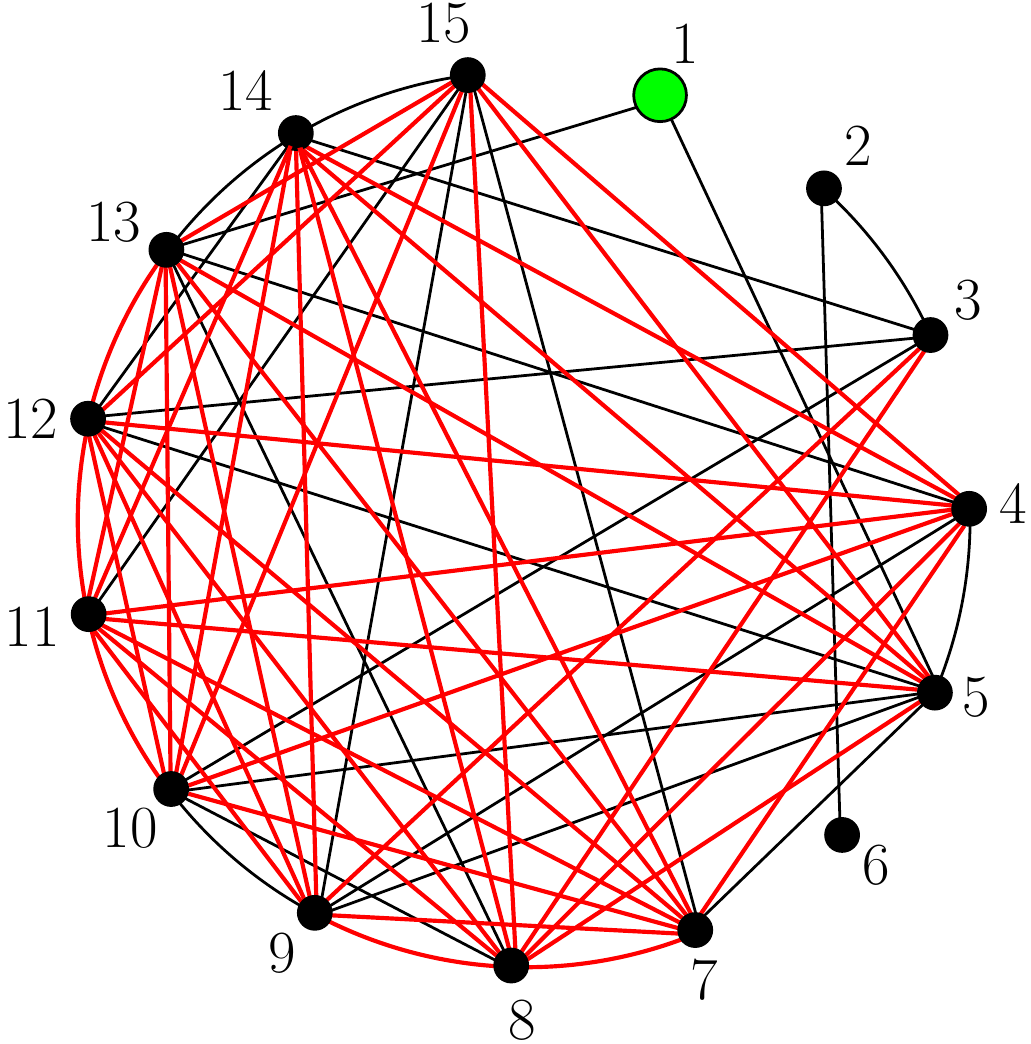}
	\caption{$\calV_\ell = \{v_1\}$}
	\end{subfigure}
	\begin{subfigure}[b]{0.15\textwidth}
	\centering
	\includegraphics[scale=0.27]{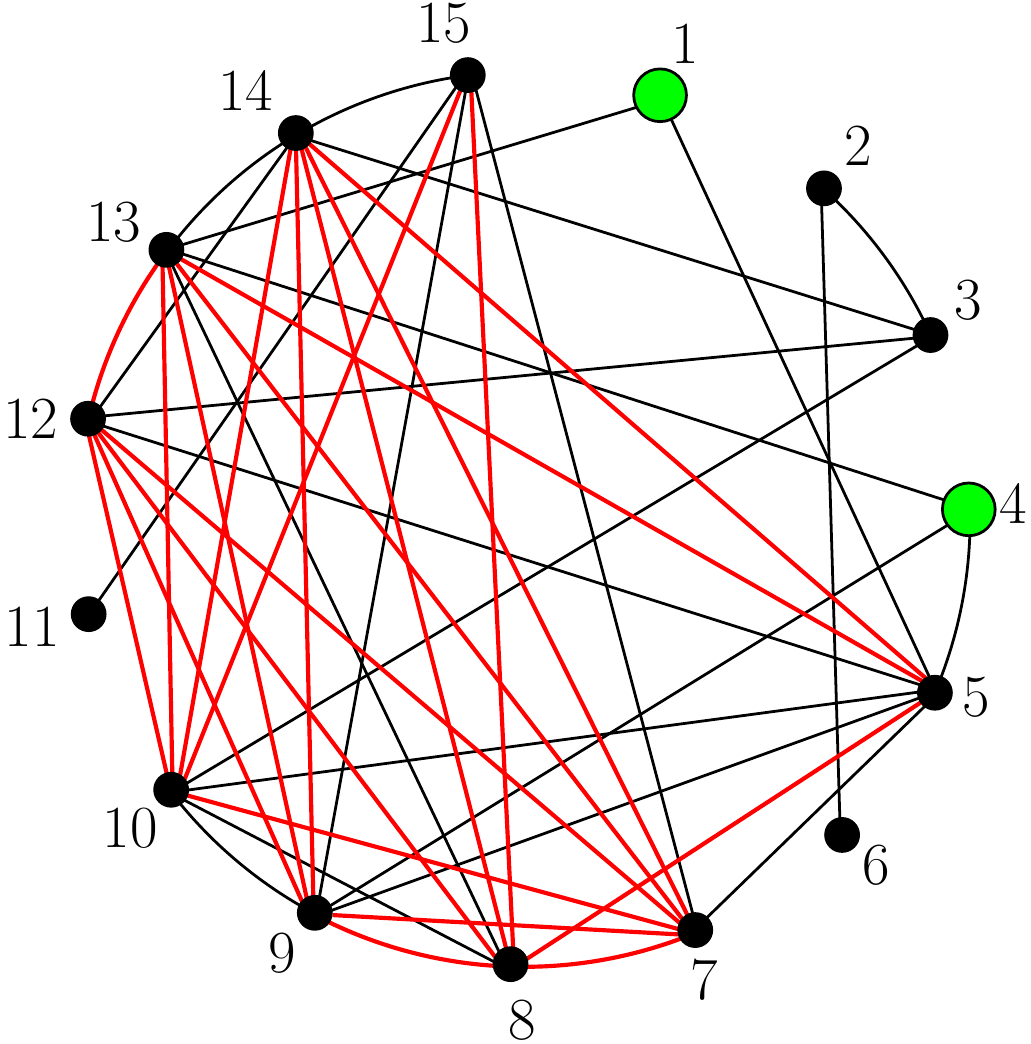}
	\caption{$\calV_\ell = \{v_1,v_4\}$}
	\end{subfigure}
\caption{Extra edges that can be added to the original graph with one and two leaders while preserving a lower bound on the dimension of SSC. Edges that are added are colored red.}
\label{fig:edge_addition}
\end{figure}

\subsection{Randomized Algorithm}
\label{sec:RA_algo}
Next, we present a randomized algorithm to add edges while preserving distances between certain node pairs in a graph. Let $\calE^c$ be the set of edges not included in the original graph $\calG=(\calV,\calE)$. In other words, $\calE^c\cup\calE$ induces a complete graph. The main idea is to randomly select an edge in $\calE^c$ and add it to the existing graph if its addition does not decreases distances between nodes in the desired node pairs. We outline the algorithm below. As previously, $\calD$ is a PMI sequence and $\bar{\calV}\subseteq\calV$ is the set of nodes whose distance-to-leader vectors are included in $\calD$.

\begin{algorithm}
\caption{Randomized Algorithm}
\label{algo:RA}
\begin{algorithmic}[1]
\State \textbf{Given} $\mathcal{G}=(\calV,\calE)$, $\mathcal{V}_\ell$, $\mathcal{D}$, $\bar{\calV}$
\State \textbf{Initialize} $\calE' \gets \calE$
\State Compute $\calE^c$ .
\State \textbf{While} $\calE^c\ne \emptyset$
\State Randomly select $e\in\calE^c$, and obtain $\calH=(\calV,\calE'\cup\{e\})$.
\State Compute $d_\calH(\ell,v)$ for all $\ell\in\calV_\ell$ and for all $v\in\bar{\calV}$.  
\State If $(d_\calH(\ell,v) = d_\calG(\ell,v))$ for all $\ell\in\calV_\ell,$ and for all $v\in\bar{\calV}$, then
$$
\calE' \gets \calE' \cup\{e\}.
$$
\State Update $\calE^c\gets \calE^c\setminus\{e\}$.
\State  \textbf{End While}
\State \textbf{Return} $\calE'$
\end{algorithmic}
\end{algorithm}
\subsubsection{Analysis}
In this subsection we analyze the performance of the randomized algorithm. 
Let $T\le|\calE^c|$ be the number of edges that are (individually) legal to add to the input graph and let $\tau \le  T$ be the size of an (unknown) optimal solution that is, the maximum number of edges from the legal set that can actually be added to the graph. Algorithm \ref{algo:RA} randomly picks $\tau' \le \tau$ legal edges, one at a time, to add to the graph. {Next, we show that if the randomized algorithm is repeated $c$ times, then with a certain probability we get a solution that is within a factor of $\alpha < 1$ of the optimal solution.} 
\begin{prop}
\label{lemma:randomized_analysis}
Algorithm~\ref{algo:RA} returns an {$\alpha$-approximate solution} with probability at least $\left(1 - e^{-c \left( \frac{\tau}{T}\right)^{{\alpha\tau}} }\right)$, when repeated $c$ times. 
\end{prop}

\emph{Proof:} The probability of the event that algorithm picks a set of legal edges that is as big as the optimal set (that is when $\tau'=\tau$) is at least $\frac{1}{ {T\choose{\tau}}}$. Let $\mathbb{S}$ be the event that algorithm finds a set of legal edges that is at least a fixed constant {$\alpha$ times the size of optimal set}, then the probability of such an event can be bounded as below.

\begin{equation}
\label{eq:prob_new}
Prob(\mathbb{S}) \ge \frac{ {\tau\choose{{\alpha\tau}}} }{ {T\choose{{\alpha\tau}}} } \ge \left( \frac{\tau}{T}\right)^{{\alpha\tau}}.
\end{equation}

This implies that the probability of failure to find a good enough set can be bounded above as

\begin{equation}
\label{eq:prob_fail}
Prob( \overline{\mathbb{S}} ) \le 1 - \left( \frac{\tau}{T}\right)^{{\alpha\tau}}.
\end{equation}

Now if we repeat our algorithm $c$ times, the probability of success improves significantly. Owing to the independence of these trials, we get

\begin{equation}
\label{eq:prob_success}
Prob(\mathbb{S}) \ge 1 - \left(1 - \left( \frac{\tau}{T}\right)^{{\alpha\tau}}\right)^c \approx 1 - e^{- \left( \frac{\tau}{T}\right)^{{\alpha\tau}} c},
\end{equation} 
which is the desired result.\qed

We note that \eqref{eq:prob_success} provides a rather loose bound on the probability of success of finding an $\alpha$-approximate solution when we repeat randomized algorithm $c$ times. We expect many {suboptimal} solutions of required size to exist that are not subsets of one fixed optimal solution. In any case, this bound is still helpful in guessing a reasonable value of $c$. {For instance, if there are $T = 100$ individually legal edges,  $\tau$ is $0.92T$, and we are interested in a solution that is at least $\alpha = 3/4$ times the size of optimal solution, then setting $c>500$ will give us a good chance (success probability of at least $0.8$) of finding a solution of required performance guarantee.}

Regarding the runtime of Algorithm~\ref{algo:RA}, note that in each iteration we need to compute distances from each leader to all nodes in the PMI sequence. Since there are at most $|\calE|$ iterations, Algorithm~\ref{algo:RA} runs in $O(|\calE|\times |\calV_\ell|\times (|\calV|+|\calE|)  )$ time. 
\section{Numerical Evaluation}
\label{sec:evaluation}
\begin{figure*}[htb]
    \centering
    \begin{subfigure}[b]{0.22\textwidth}
	\centering
	\includegraphics[scale=0.27]{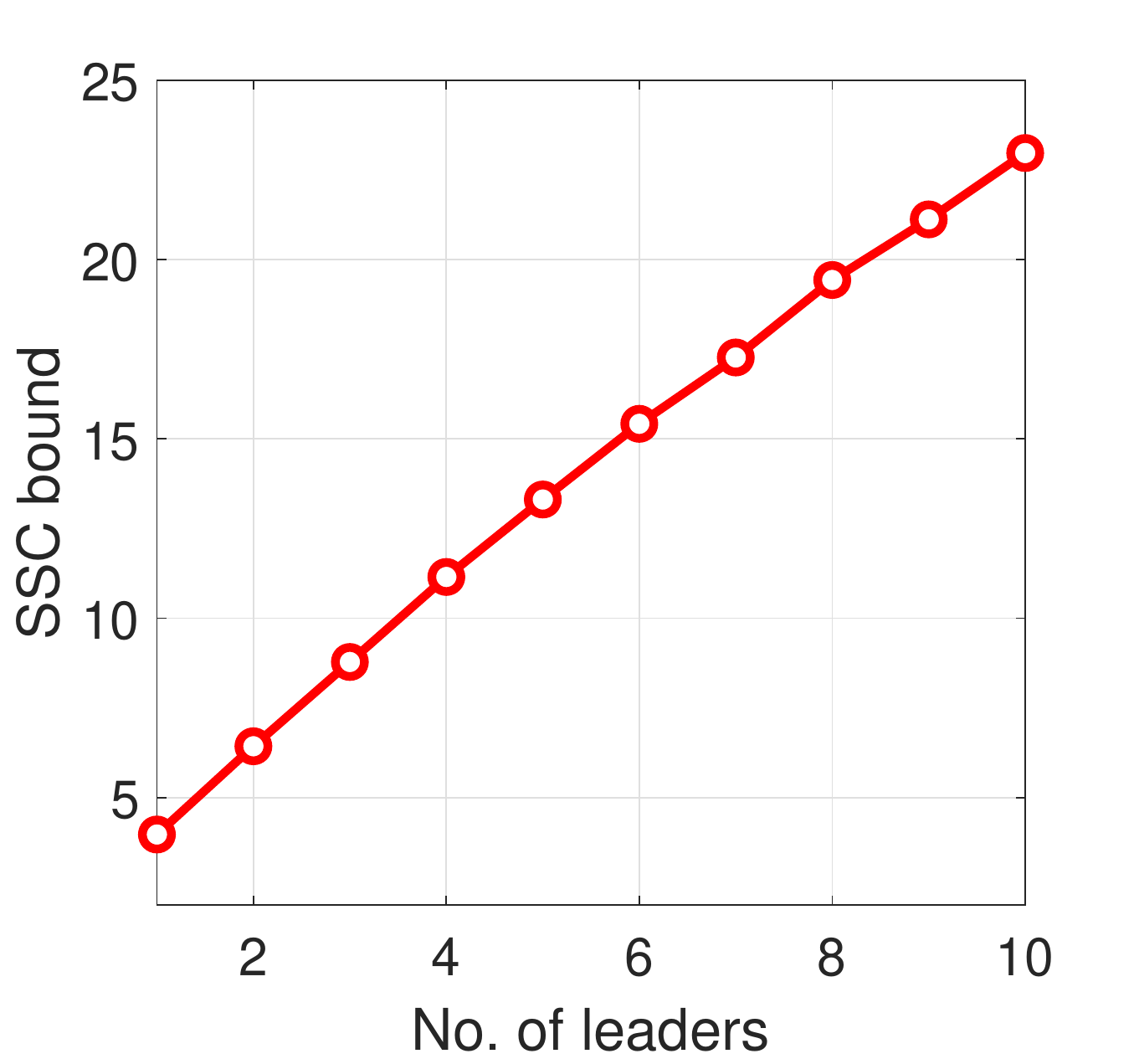}
	\caption{$p=0.2$}
	\end{subfigure}
	 \begin{subfigure}[b]{0.22\textwidth}
	\centering
	\includegraphics[scale=0.27]{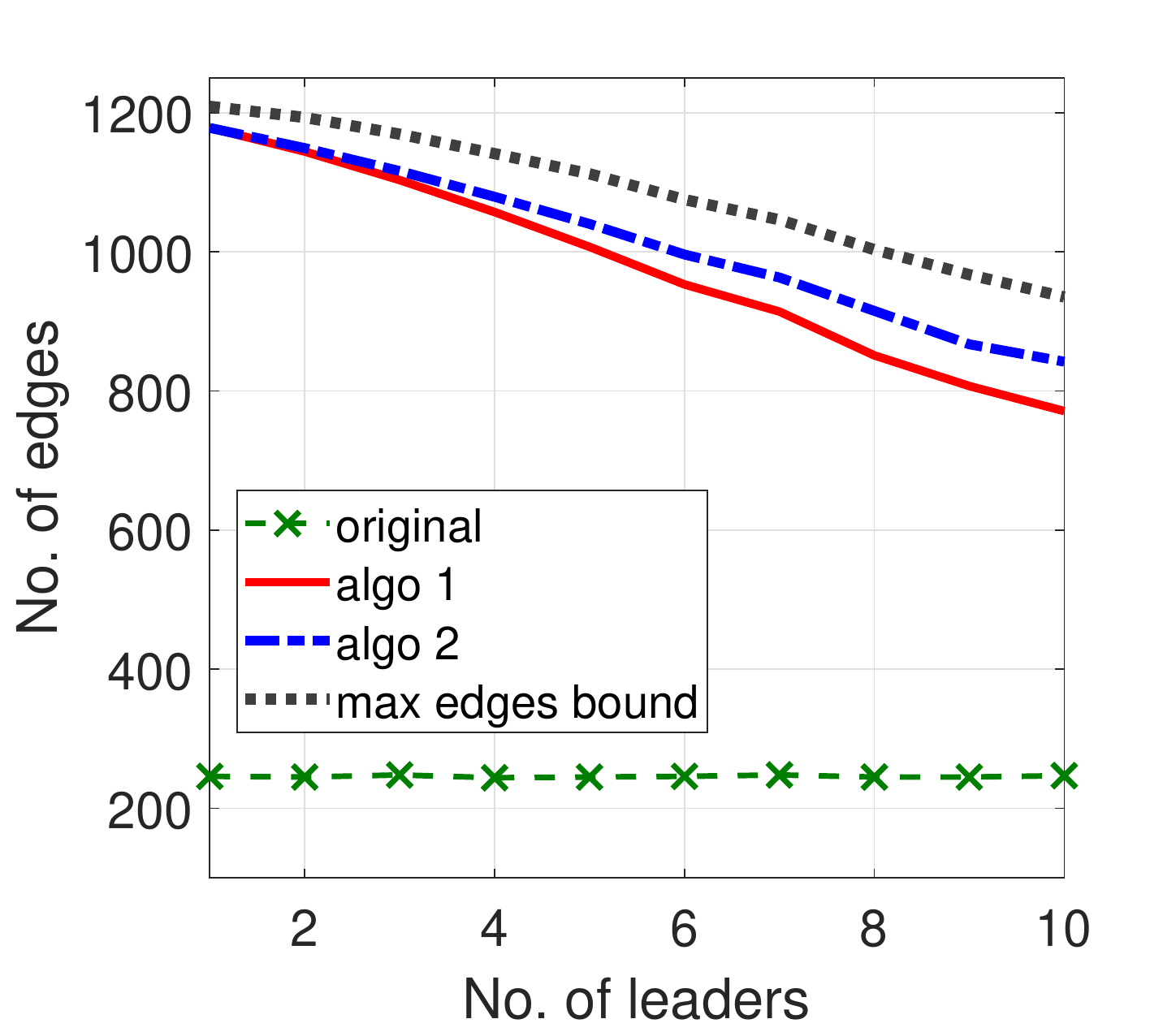}
	\caption{$p=0.2$}
	\end{subfigure}
	\begin{subfigure}[b]{0.22\textwidth}
	\centering
	\includegraphics[scale=0.27]{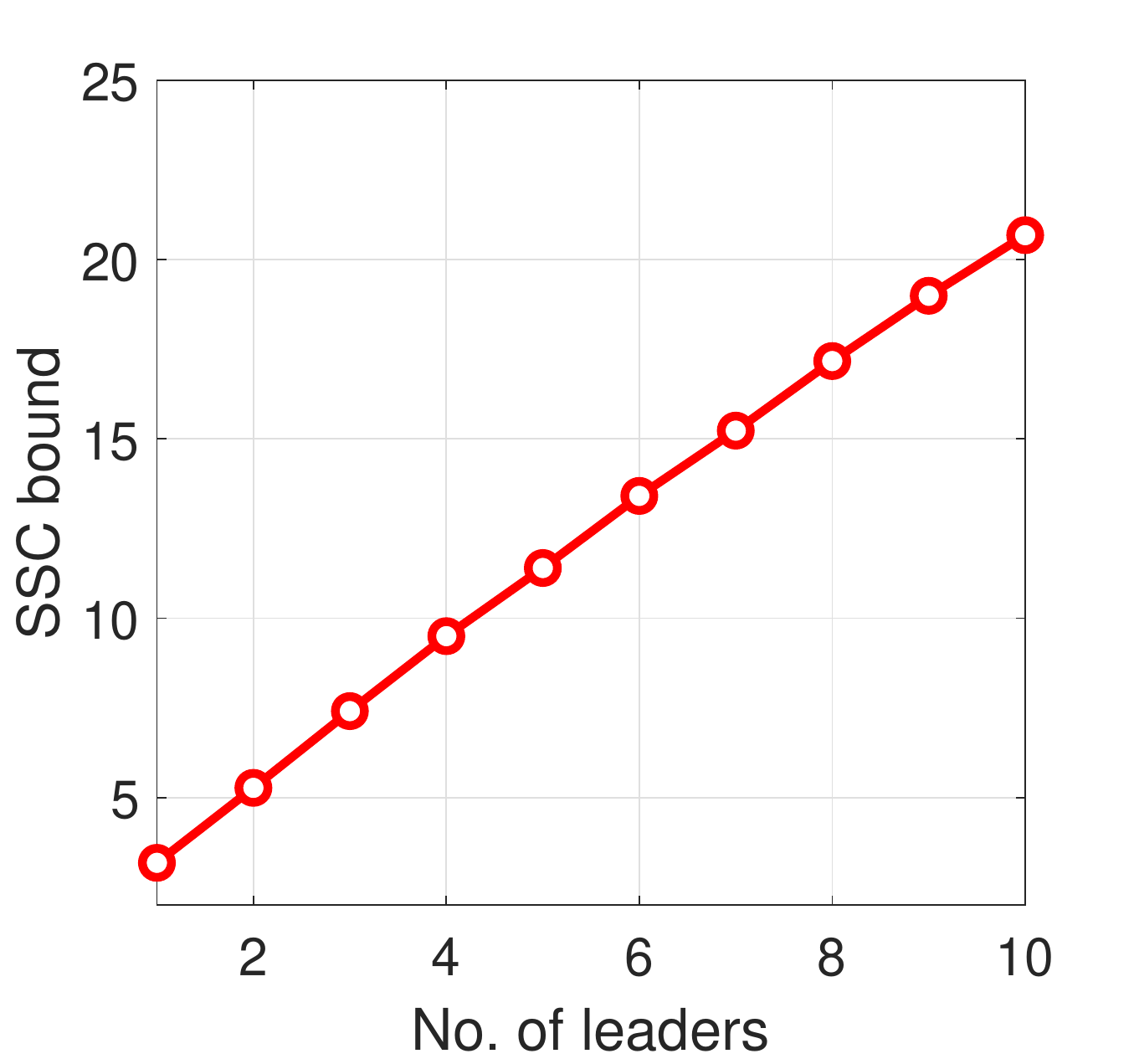}
	\caption{$p=0.3$}
	\end{subfigure}
	 \begin{subfigure}[b]{0.22\textwidth}
	\centering
	\includegraphics[scale=0.27]{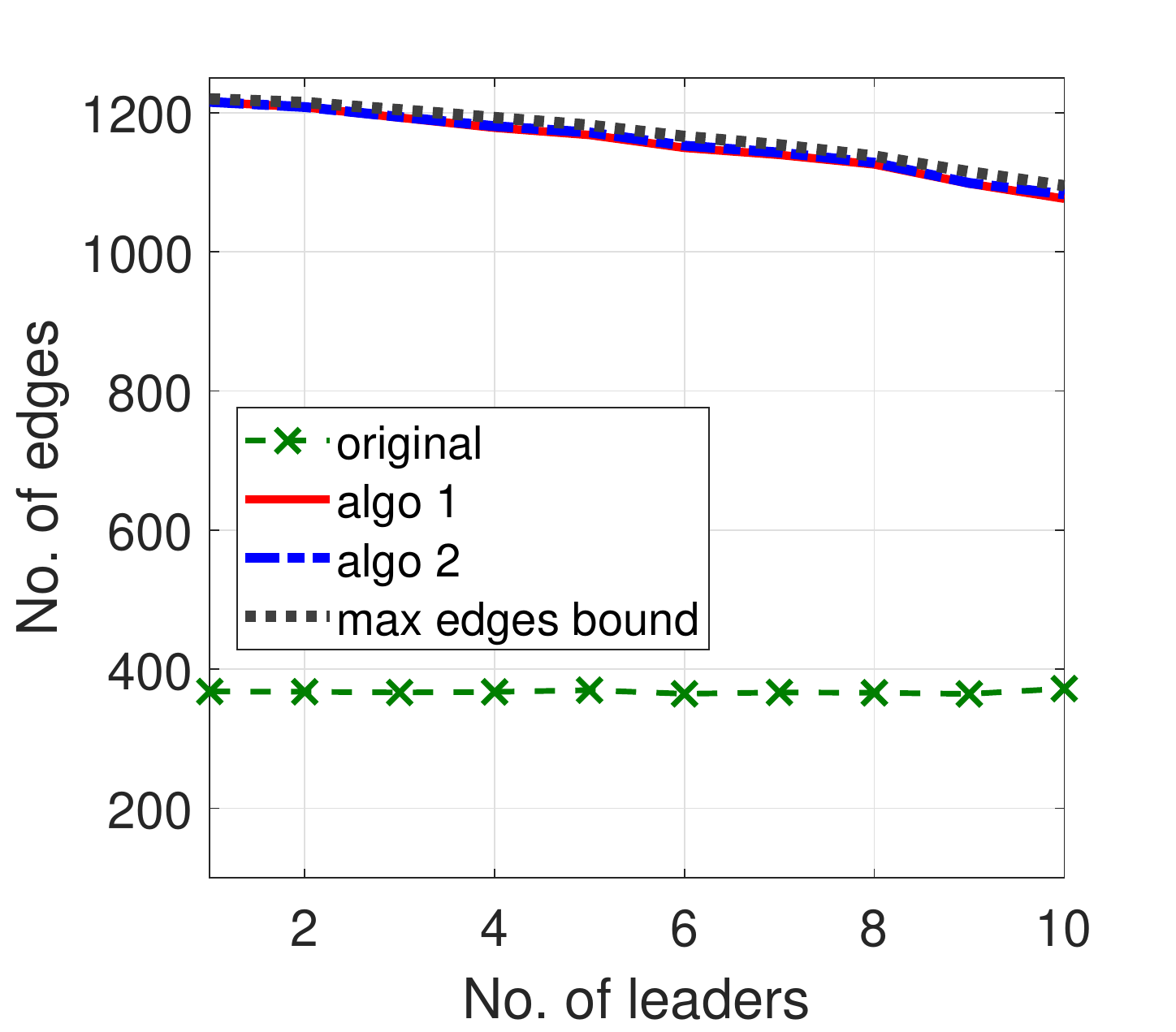}
	\caption{$p=0.3$}
	\end{subfigure}
    \caption{\emph{ER--networks:} For $p=0.2$ and $0.3$ respectively, (a) and (c) illustrate lower bound on the dimension of SSC as a function of number of leaders. In (b) and (d), Algorithms 1 and 2 are compared. 
    Edges in original graphs as well as upper bound on the maximum number of edges that can be added without changing PMI sequence are also shown.}
    \label{fig:ER_plots}
\end{figure*}

\begin{figure*}[htb]
    \centering
    \begin{subfigure}[b]{0.22\textwidth}
	\centering
	\includegraphics[scale=0.27]{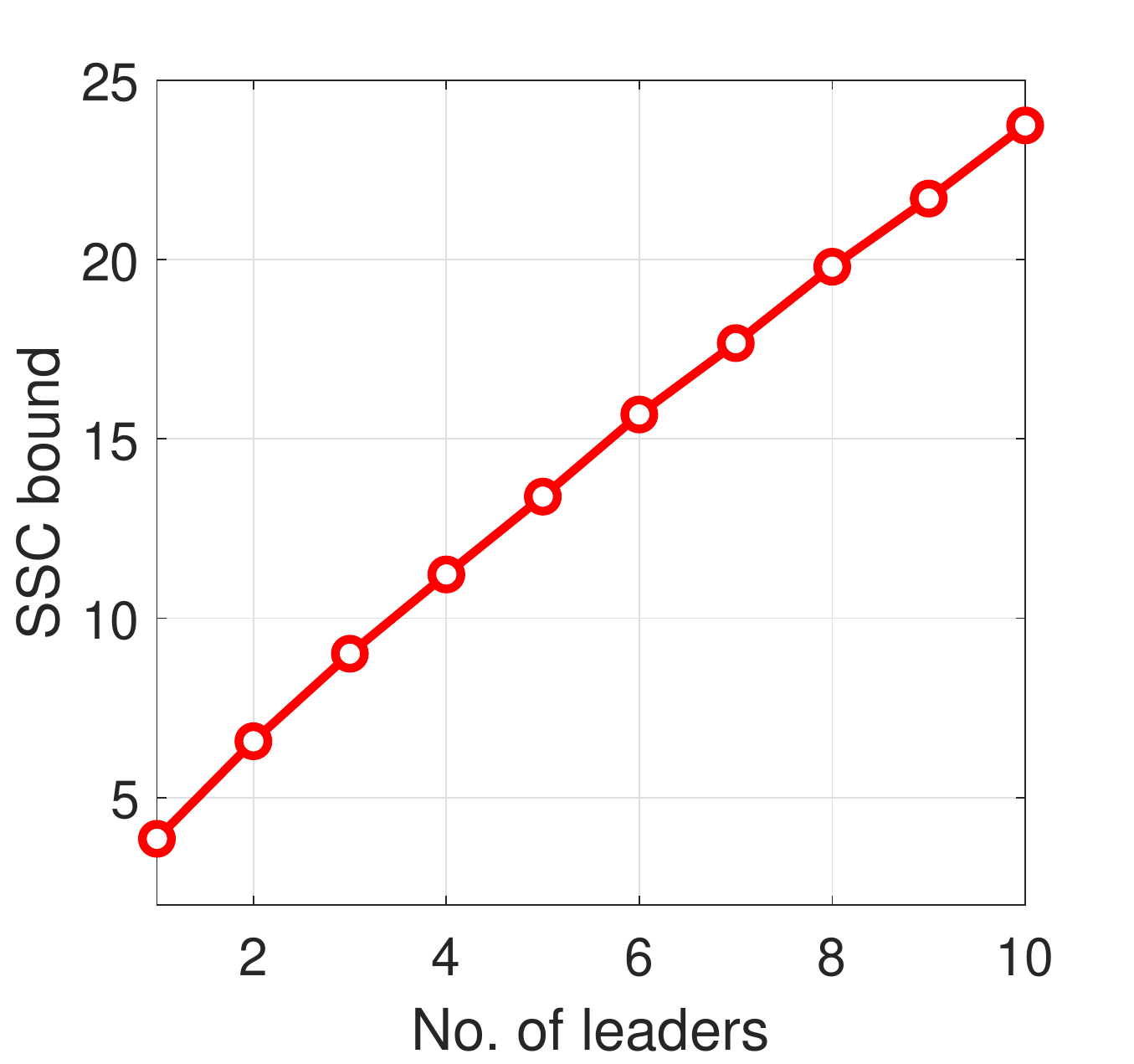}
	\caption{$\gamma=5$}
	\end{subfigure}
	 \begin{subfigure}[b]{0.22\textwidth}
	\centering
	\includegraphics[scale=0.27]{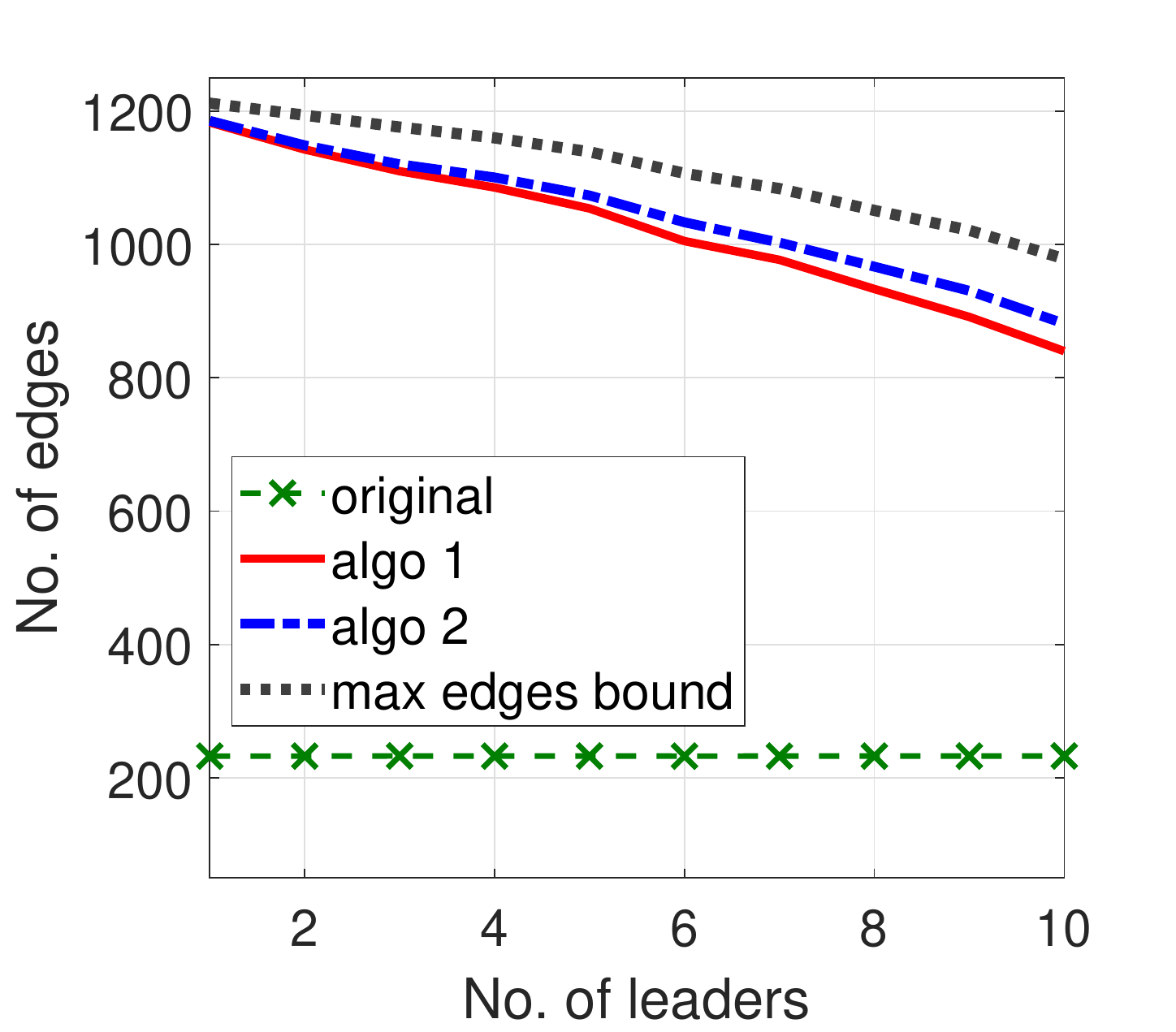}
	\caption{$\gamma=5$}
	\end{subfigure}
	\begin{subfigure}[b]{0.22\textwidth}
	\centering
	\includegraphics[scale=0.27]{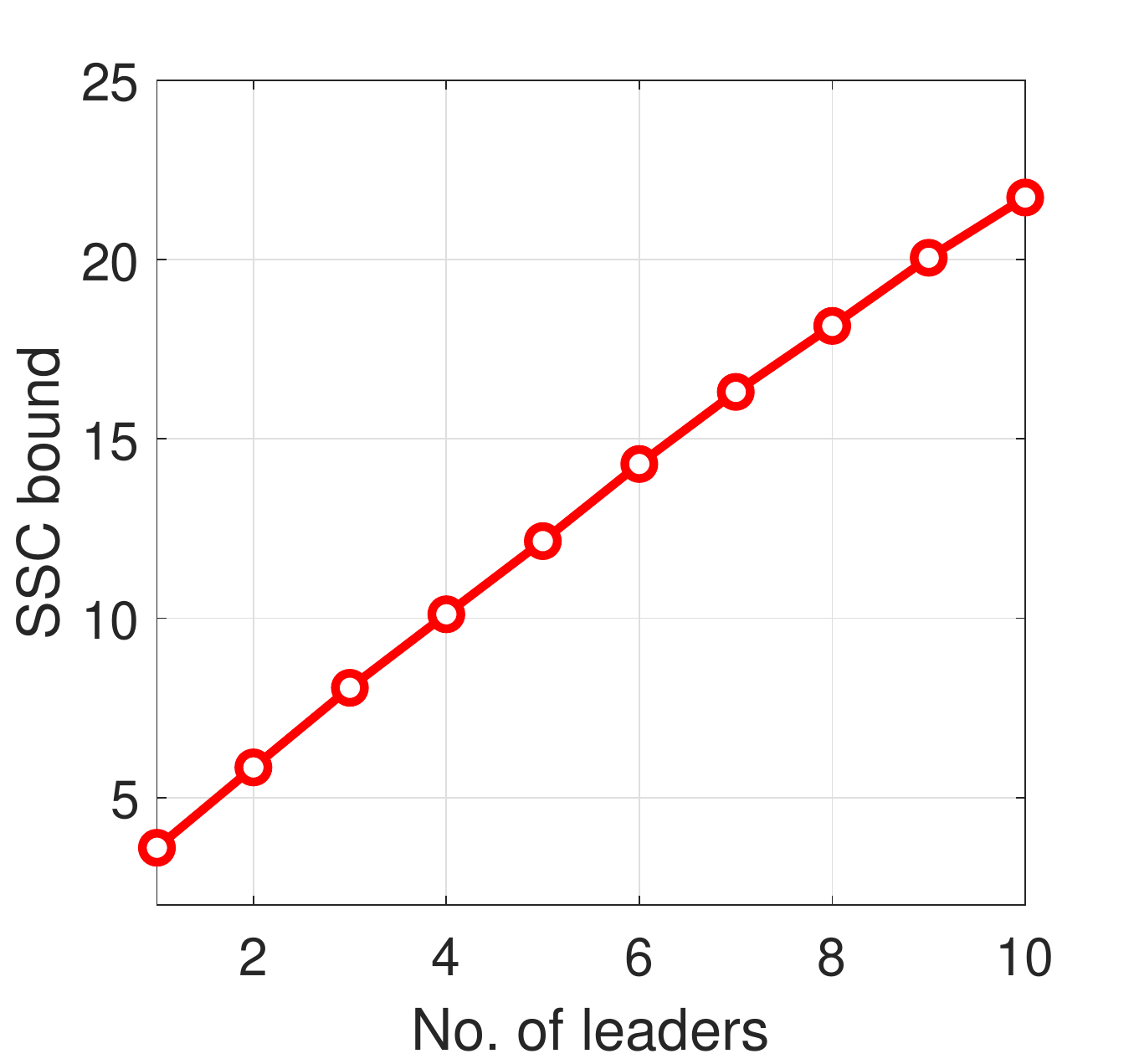}
	\caption{$\gamma=7$}
	\end{subfigure}
	 \begin{subfigure}[b]{0.22\textwidth}
	\centering
	\includegraphics[scale=0.27]{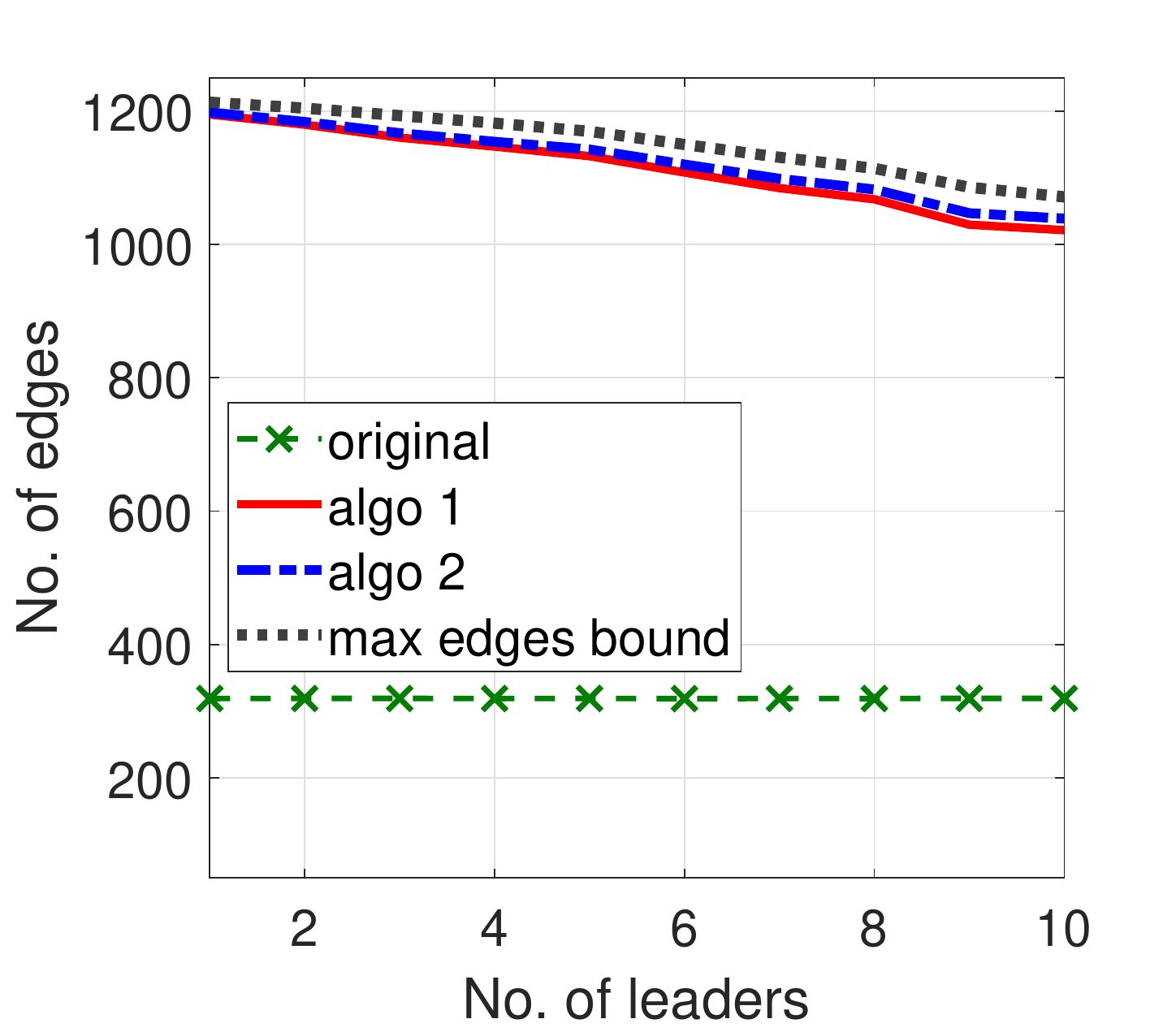}
	\caption{$\gamma=7$}
	\end{subfigure}
    \caption{\emph{BA--Networks:} For $\gamma=5$ and $7$ respectively, (a) and (c) illustrate bound on dimension of SSC as a function of number of leaders. In (b) and (d), a comparison of Algorithms 1 and 2 is illustrated.}
    \label{fig:BA_plots}
\end{figure*}
Here, we evaluate our algorithms on Erd\H{o}s-R\'{e}nyi (ER) networks in which any two nodes are adjacent with probability $p$, and Barab\'{a}si-Albert (BA) networks in which each new node is adjacent to $\gamma$ existing nodes through a preferential attachment strategy. For both ER and BA models, we consider networks of 50 nodes.

Figure \ref{fig:ER_plots} illustrates results for ER graphs. For a selected $p$, first we plot distance-based lower bound on the dimension of SSC as a function of the number of leaders selected randomly (Figures \ref{fig:ER_plots}(a) and \ref{fig:ER_plots}(c)). Then, using Algorithms 1 and 2, we add edges to networks while preserving lower bounds on their dimensions of SSC. We also compare results against an upper bound on the optimal number of edges that can be added without changing PMI sequence. The bound is obtained by observing that an edge cannot be added between two such nodes that lie on a shortest path between a leader and a node whose distance-to-leader vector is included in a given PMI sequence. This observation gives an upper bound on the maximum number of edges that can be added in a graph while preserving a PMI sequence. 

Plots in Figures \ref{fig:ER_plots}(b) and \ref{fig:ER_plots}(d) indicate that new networks obtained after applying Algorithms \ref{algo:DPEA} and \ref{algo:RA} contain significantly more edges as compared to the original network. Algorithm 2 performs slightly better than the Algorithm 1, especially when the number of leaders increases. However, as the value of $p$ increases, the difference between two algorithms is negligible. To obtain results of Algorithm 2, we perform 150 repetitions (that is $c = 150$) and then select the best solution. We also mention that Algorithm 1 takes significantly less time as compared to the Algorithm 2 (based on the choice of $c$). Similar results are obtained for BA networks as illustrated in Figure \ref{fig:BA_plots}. In all the plots, value at each point is an average of 100 randomly generated instances.
\section{Conclusion}
\label{sec:Con}
Adding extra links between nodes improves network's robustness to noisy information and to structural changes in the underlying network topology, but at the same time, these extra links can deteriorate network controllability. We introduced the problem of improving network robustness by adding edges while maintaining a lower bound on the dimension of SSC. By exploring the relationship between strong structural controllability and distances between nodes in a graph, we showed that the above problem can be formulated as an edge augmentation problem with the constraint of maintaining distances between a certain pair of nodes.  To solve the edge augmentation problem, we presented deterministic and randomized algorithms that approximately computed densest graphs preserving a given set of pairwise distances between nodes. 
We believe that characterizing densest graphs that preserve distances between certain node pairs is an interesting problem in its own respect. We aim to study this problem further and design efficient exact algorithms to solve it. 
\bibliographystyle{IEEEtran}
\bibliography{references}

\begin{thebibliography}{10}
\providecommand{\url}[1]{#1}
\csname url@samestyle\endcsname
\providecommand{\newblock}{\relax}
\providecommand{\bibinfo}[2]{#2}
\providecommand{\BIBentrySTDinterwordspacing}{\spaceskip=0pt\relax}
\providecommand{\BIBentryALTinterwordstretchfactor}{4}
\providecommand{\BIBentryALTinterwordspacing}{\spaceskip=\fontdimen2\font plus
\BIBentryALTinterwordstretchfactor\fontdimen3\font minus
  \fontdimen4\font\relax}
\providecommand{\BIBforeignlanguage}[2]{{%
\expandafter\ifx\csname l@#1\endcsname\relax
\typeout{** WARNING: IEEEtran.bst: No hyphenation pattern has been}%
\typeout{** loaded for the language `#1'. Using the pattern for}%
\typeout{** the default language instead.}%
\else
\language=\csname l@#1\endcsname
\fi
#2}}
\providecommand{\BIBdecl}{\relax}
\BIBdecl

\bibitem{pasqualetti2014controllability}
F.~Pasqualetti, S.~Zampieri, and F.~Bullo, ``Controllability metrics,
  limitations and algorithms for complex networks,'' \emph{IEEE Transactions on
  Control of Network Systems}, vol.~1, no.~1, pp. 40--52, 2014.

\bibitem{chapman2013strong}
A.~Chapman and M.~Mesbahi, ``On strong structural controllability of networked
  systems: A constrained matching approach,'' in \emph{American Control
  Conference}, 2013, pp. 6126--6131.

\bibitem{monshizadeh2014zero}
N.~Monshizadeh, S.~Zhang, and M.~K. Camlibel, ``Zero forcing sets and
  controllability of dynamical systems defined on graphs,'' \emph{IEEE
  Transactions on Automatic Control}, vol.~59, pp. 2562--2567, 2014.

\bibitem{aguilar2015graph}
C.~O. Aguilar and B.~Gharesifard, ``Graph controllability classes for the
  {L}aplacian leader-follower dynamics,'' \emph{IEEE Transactions on Automatic
  Control}, vol.~60, no.~6, pp. 1611--1623, 2015.

\bibitem{yaziciouglu2016graph}
A.~Yaz{\i}c{\i}o{\u{g}}lu, W.~Abbas, and M.~Egerstedt, ``Graph distances and
  controllability of networks,'' \emph{IEEE Transactions on Automatic Control},
  vol.~61, no.~12, pp. 4125--4130, 2016.

\bibitem{mousavi2016controllability}
S.~S. Mousavi and M.~Haeri, ``Controllability analysis of networks through
  their topologies,'' in \emph{2016 IEEE 55th Conference on Decision and
  Control (CDC)}.\hskip 1em plus 0.5em minus 0.4em\relax IEEE, 2016, pp.
  4346--4351.

\bibitem{van2017distance}
H.~J. Van~Waarde, M.~K. Camlibel, and H.~L. Trentelman, ``A distance-based
  approach to strong target control of dynamical networks,'' \emph{IEEE
  Transactions on Automatic Control}, vol.~62, no.~12, 2017.

\bibitem{young2010robustness}
G.~F. Young, L.~Scardovi, and N.~E. Leonard, ``Robustness of noisy consensus
  dynamics with directed communication,'' in \emph{Proceedings of the American
  Control Conference}, 2010, pp. 6312--6317.

\bibitem{ellens2011effective}
W.~Ellens, F.~Spieksma, P.~Van~Mieghem, A.~Jamakovic, and R.~Kooij, ``Effective
  graph resistance,'' \emph{Linear Algebra and its Applications}, vol. 435,
  no.~10, pp. 2491--2506, 2011.

\bibitem{abbas2012robust}
W.~Abbas and M.~Egerstedt, ``Robust graph topologies for networked systems,''
  in \emph{3rd IFAC Workshop on Distributed Estimation and Control in Netwoked
  Systems}, 2012, pp. 85--90.

\bibitem{zelazo2015robustness}
D.~Zelazo and M.~B{\"u}rger, ``On the robustness of uncertain consensus
  networks,'' \emph{IEEE Transactions on Control of Network Systems}, vol.~4,
  no.~2, pp. 170--178, 2015.

\bibitem{YasinCDC19}
Y.~S. Yazicioglu, W.~Abbas, and Mudassir, ``Structural robustness to noise in
  consensus networks: Impact of average degrees and average distances,'' in
  \emph{58th IEEE Conference on Decision and Control}, 2019.

\bibitem{Abbas2019}
W.~Abbas, M.~Shabbir, A.~Yaz{\i}c{\i}o{\u{g}}lu, and A.~Akbar, ``On the
  trade-off between controllability and robustness in networks of diffusively
  coupled agents,'' in \emph{American Control Conference (ACC)}, 2019.

\bibitem{pasqualetti2018fragility}
F.~Pasqualetti, C.~Favaretto, S.~Zhao, and S.~Zampieri, ``Fragility and
  controllability tradeoff in complex networks,'' in \emph{American Control
  Conference (ACC)}, 2018, pp. 216--221.

\bibitem{weber2014linear}
A.~Weber, G.~Reissig, and F.~Svaricek, ``A linear time algorithm to verify
  strong structural controllability,'' in \emph{53rd IEEE Conference on
  Decision and Control}, 2014, pp. 5574--5580.

\bibitem{monshizadeh2015strong}
N.~Monshizadeh, K.~Camlibel, and H.~Trentelman, ``Strong targeted
  controllability of dynamical networks,'' in \emph{54th IEEE Conference on
  Decision and Control (CDC)}, 2015, pp. 4782--4787.

\bibitem{zhang2014upper}
S.~Zhang, M.~Cao, and M.~K. Camlibel, ``Upper and lower bounds for controllable
  subspaces of networks of diffusively coupled agents,'' \emph{IEEE
  Transactions on Automatic Control}, vol.~59, pp. 745--750, 2014.

\bibitem{mousavi2018structural}
S.~S. Mousavi, M.~Haeri, and M.~Mesbahi, ``On the structural and strong
  structural controllability of undirected networks,'' \emph{IEEE Transactions
  on Automatic Control}, vol.~63, no.~7, pp. 2234--2241, 2018.

\bibitem{bernstein2019distance}
A.~Bernstein, K.~Däubel, Y.~Disser, M.~Klimm, T.~Mütze, and F.~Smolny,
  ``Distance-preserving graph contractions,'' \emph{SIAM Journal on Discrete
  Mathematics}, vol.~33, no.~3, pp. 1607--1636, 2019.

\bibitem{baswana2007simple}
S.~Baswana and S.~Sen, ``A simple and linear time randomized algorithm for
  computing sparse spanners in weighted graphs,'' \emph{Random Structures \&
  Algorithms}, vol.~30, no.~4, pp. 532--563, 2007.

\bibitem{bollobas2005sparse}
B.~Bollob{\'a}s, D.~Coppersmith, and M.~Elkin, ``Sparse distance preservers and
  additive spanners,'' \emph{SIAM Journal on Discrete Mathematics}, vol.~19,
  no.~4, pp. 1029--1055, 2005.

\bibitem{spielman2011graph}
D.~A. Spielman and N.~Srivastava, ``Graph sparsification by effective
  resistances,'' \emph{SIAM Journal on Computing}, vol.~40, no.~6, pp.
  1913--1926, 2011.

\bibitem{mousavi2017robust}
S.~S. Mousavi, M.~Haeri, and M.~Mesbahi, ``Robust strong structural
  controllability of networks with respect to edge additions and deletions,''
  in \emph{American Control Conference (ACC)}, 2017, pp. 5007--5012.

\bibitem{pequito2015complexity}
S.~Pequito, S.~Kar, and A.~P. Aguiar, ``On the complexity of the constrained
  input selection problem for structural linear systems,'' \emph{Automatica},
  vol.~62, pp. 193--199, 2015.

\bibitem{tzoumas2015minimal}
V.~Tzoumas, M.~A. Rahimian, G.~J. Pappas, and A.~Jadbabaie, ``Minimal actuator
  placement with bounds on control effort,'' \emph{IEEE Transactions on Control
  of Network Systems}, vol.~3, pp. 67--78, 2015.

\bibitem{pequito2015framework}
S.~Pequito, S.~Kar, and A.~P. Aguiar, ``A framework for structural input/output
  and control configuration selection in large-scale systems,'' \emph{IEEE
  Transactions on Automatic Control}, vol.~61, pp. 303--318, 2015.

\bibitem{olshevsky2014minimal}
A.~Olshevsky, ``Minimal controllability problems,'' \emph{IEEE Transactions on
  Control of Network Systems}, vol.~1, no.~3, pp. 249--258, 2014.

\bibitem{yaziciouglu2013leader}
A.~Y. Yazicio{\u{g}}lu and M.~Egerstedt, ``Leader selection and network
  assembly for controllability of leader-follower networks,'' in \emph{2013
  American Control Conference}.\hskip 1em plus 0.5em minus 0.4em\relax IEEE,
  2013, pp. 3802--3807.

\bibitem{clark2017submodularity}
A.~Clark, B.~Alomair, L.~Bushnell, and R.~Poovendran, ``Submodularity in input
  node selection for networked linear systems: Efficient algorithms for
  performance and controllability,'' \emph{IEEE Control Systems Magazine},
  vol.~37, no.~6, pp. 52--74, 2017.

\bibitem{MudassirCDC19}
\BIBentryALTinterwordspacing
M.~Shabbir, W.~Abbas, and Y.~Yazicioglu, ``On the computation of a lower bound
  on strong structural controllability in networks,'' in \emph{58th IEEE
  Conference on Decision and Control}, 2019. [Online]. Available:
  \url{https://arxiv.org/abs/1909.03565}
\BIBentrySTDinterwordspacing

\end{thebibliography}
\end{document}